\documentclass[12pt]{article}
\usepackage[height=220mm,width=150mm]{geometry}
\usepackage{array}
\usepackage{color}
\usepackage{mathrsfs}
\usepackage{hyperref}
\usepackage{amssymb}
\usepackage{amsmath}
\usepackage{amsthm}
\usepackage{color}
\usepackage{enumitem}
\allowdisplaybreaks
\usepackage{graphicx}
\usepackage{subfigure}
\usepackage{tikz}
\usetikzlibrary{shapes.arrows,chains,positioning}
\newtheorem{Theorem}{Theorem}[section]
\newtheorem{Lemma}[Theorem]{Lemma}
\newtheorem{Definition}[Theorem]{Definition}
\newtheorem{Proposition}[Theorem]{Proposition}

\newtheorem{Remark}[Theorem]{Remark}
\usepackage{amssymb}

\def\nocolor#1{}

\begin{document}


\title{Uniqueness of phase retrieval from offset linear canonical transform}

\author{Jing Liu and Haiye Huo\thanks{Corresponding author.}\\
\normalsize{Department of Mathematics, School of Mathematics and Computer Sciences},\\ \normalsize{Nanchang University, Nanchang~330031, Jiangxi, China} \\
\normalsize{E-mail: liujing5965@163.com, hyhuo@ncu.edu.cn}}

\date{}
\maketitle

\textbf{Abstract:}
The classical phase retrieval refers to the recovery of an unknown signal from its Fourier magnitudes, which is widely used in fields such as quantum mechanics, signal processing, optics, etc. The offset linear canonical transform (OLCT), which is a more general type of linear integral transform including Fourier transform (FT), fractional Fourier transform (FrFT), and linear canonical transform (LCT) as its special cases. Hence, in this paper, we focus on the uniqueness problem of phase retrieval in the framework of OLCT. First, we prove that all the nontrivial ambiguities in continuous OLCT phase retrieval can be represented by convolution operators, and demonstrate that a continuous compactly supported signal can be uniquely determined up to a global phase from its multiple magnitude-only OLCT measurements. Moreover, we investigate the nontrivial ambiguities in the discrete OLCT phase retrieval case. Furthermore, we demenstrate that a nonseparable function can be uniquely recovered from its magnitudes of short-time OLCT (STOLCT) up to a global phase. Finally, we show that signals which are bandlimited in FT or OLCT domain can be reconstructed from its sampled STOLCT magnitude measurements, up to a global phase, providing the ambiguity function of window function satisfies some mild conditions.



\section{Introduction}\label{sec:I}

The problem of recovering an original signal from its magnitudes of FT measurements is called phase retrieval \cite{BCM14,YL24,BCE06,CS21,QT16}. It is widely used in optical imaging \cite{Seifert2004}, astronomy \cite{Bruck1979}, X-ray crystallography \cite{millane1990}, signal processing \cite{CSV13}, and quantum mechanics \cite{CMS01}, etc. As we all known, phase retrieval problem is ill posed, and an unknown signal cannot be reconstructed from its Fourier amplitude, due to trivial ambiguities. The non-uniqueness of classical phase retrieval has been well studied, and can be avoided by imposing some prior information or adding additional measurements. For more details, we refer the reader to references \cite{LL21,KK14,ADG19,BP15,CCS22,SEC15}.

From a physics perspective, the magnitude of FT can be seen as the intensity of a wave in the far field. However, it is very challenging to establish a unified framework that can effectively and accurately calculate near-field diffraction. Note that the near-field diffraction can be characterized by Fresnel transform or fractional Fourier transform (FRFT). Hence, in this paper, we study the phase retrieval problem for a more general linear integrable transform--OLCT \cite{Abe1994}. The OLCT has five free parameters, which offer more freedom and flexibility in real applications. Many famous linear integrable transforms, such as FT, Fresnel transform, FRFT \cite{Almeida1994}, LCT \cite{Barshan1997}, are all special cases of the OLCT. Therefore, it is meaningful to investigate the phase retrieval problem under the framework of OLCT. However, to the best of our knowledge, there are few studies on phase retrieval for OLCT. To fill up this gap, in this paper, we consider the OLCT phase retrieval problem and prove that signals can be uniquely recovered from the magnitude of OLCT (or STOLCT) measurements, up to a global phase, respectively.

Since there exists some trivial ambiguities in the phase retrieval problem, such as rotation, shift and conjugate reflection, some prior information or additional measurements on the function are required to ensure that the original signal can be uniquely recovered. Among them, a common approach is to suppose that the continuous-time signal is compactly supported. Beinert \cite{Beinert17} stated that an original signal $f$ and the unknown interference $h$ with compact support can be uniquely reconstructed from their Fourier magnitudes $\mathcal{F}f$, $\mathcal{F}h$, and $\mathcal{F}(f+h)$, where $\mathcal{F}$ represents the FT operator. Recently, many scholars have extended the classical phase retrieval results from the FT domain to a more generalized linear integral transform domain, such as FRFT \cite{CHS15,AJ16,STL22,YTW24}, LCT \cite{Beinert2017,ChenYang2022,LWH24}. Jaming \cite{Philippe2014} showed that a compactly supported signal can be exactly recovered from its multiple FrFT intensities, up to a global phase. Beinert \cite{Beinert2017} gave a characterization of compactly supported functions which have same LCT magnitudes. Chen et al.\cite{ChenYang2022} demonstrated that the nontrivial ambiguities can be represented as a convolution of two compactly supported functions. Unfortunately, there is no conclusion about OLCT phase retrieval for compactly supported functions. Therefore, motivated by \cite{BP15,Beinert2017,ChenYang2022}, one of our main contributions is to derive some uniqueness conditions of OLCT phase retrieval for both continuous and discrete signals that are compactly supported.

Another approach to overcome the trivial ambiguities is to obtain more information by additional magnitude-only measurements, such as short-time FT (STFT) \cite{LMJ23,QingyueZhangZhenli2023,LLZ21}. Alaifari et al \cite{AlaifariRima2021} showed that signals can be uniquely determined by their STFT magnitude up to a global phase, providing the ambiguity function of the window function vanishes only on some of the time-frequency plane. Zhang \cite{Zhang18} proposed some sufficient conditions that can guarantee the uniqueness of discrete short-time FrFT phase retrieval. Li et al \cite{LiRui2024} investigated the phase retrieval problem for short-time LCT (STLCT) and established some uniqueness results for nonseparable functions, bandlimited functions, and cardinal B-spline functions. Li et al. \cite{LZL25} introduced a generalized Paley-Wiener space, and derived some uniqueness results for a complex-valued signal that is reconstructed from its sampled STLCT magnitudes. Zhang et al. \cite{ZSL21} considered two sufficient conditions that ensure arbitrary signal can be uniquely recovered form the magnitude of discrete STOLCT measurements. However, to the best of our knowledge, there are no studies on phase retrieval for continuous STOLCT. To fill up this gap, in this paper, we consider the uniqueness problem for continuous STOLCT phase retrieval. More precisely, we extend the work of Alaifari et al \cite{AlaifariRima2021} and Li et al \cite{LiRui2024} to STOLCT domain, and show that a real-valued signal which are nonseparable can be uniquely recovered from their STOLCT magnitudes, up to a sign phase, providing the ambiguity function of the window satisfies some mild conditions. Moreover, we prove that the complex-valued function which is bandlimited in FT domain or OLCT domain, can be uniquely determined from the sampled STOLCT magnitudes, up to a global phase.

The rest of the paper is organized as follows. In Section~\ref{sec:P}, we review some preliminary results concerning properties of OLCT, STOLCT, ambiguity function, and bandlimited function. In Section~\ref{CPL}, we investigate uniqueness of OLCT phase retrieval for continuous compactly supported functions. We extend the characterization of nontrivial ambiguities to discrete OLCT in Section \ref{DPR}. In Section \ref{STPR}, we study the uniqueness of STOLCT phase retrieval for nonseparable functions and bandlimited functions, respectively. Finally, we conclude the paper in Section \ref{conl}.

\section{Preliminaries}\label{sec:P}

In this section, we summarize all the fundamental definitions used in this paper, including OLCT, STOLCT, ambiguity function, and bandlimited function, etc.

\subsection{Continuous OLCT}

\begin{Definition}\cite{Adrian2007,Huo19C,Huo19M}
The OLCT of a continuous function $f(t)\in L^2(\mathbb{R})$ with real parameters $A=\left[
      \begin{array}{cc|c}
        a & b & y_0 \\
        c & d & \omega_0 \\
      \end{array}
      \right]$
is given by
\begin{equation}\label{def:COLCT}
O^{A}f(u)=O^A\left[f(t)\right](u)=
    \begin{cases}
    \int_{-\infty}^{+\infty}f(t)K_{A}(t,u){\rm{d}}t, & b\ne0,\\
     \sqrt{d}e^{j\frac{cd}{2}(u-y_0)^2+ju\omega_0}f\left(d(u-y_0)\right), & b=0,\\
     \end{cases}
\end{equation}
where the transform kernel $K_A$ is denoted as
\begin{equation}\label{ker}
K_{A}(t,u)=\frac{1}{\sqrt{j2\pi b}}e^{\frac{j}{2b}\left[du^2+dy_0^2+at^2+2t(y_0-u)-2u(dy_0-b\omega_0)\right]},
\end{equation}
and $ad-bc=1$.
\end{Definition}

For $b=0$, the OLCT reduces to a chirp multiplication operator. Hence, without loss of generality, we assume $b>0$ throughout the paper.

Note that $O^A$ is a unitary operator on $L^2(\mathbb{R})$, and the inverse OLCT is defined by an OLCT with parameters
$A^{-1}=\left[
      \begin{array}{cc|c}
        d & -b & b\omega_0-dy_0 \\
        -c & a & cy_0-a\omega_0 \\
      \end{array}
      \right]$\cite{Huo19M,Pei2003},
i.e.,
\[
f(t)=C\int_{-\infty}^{+\infty}O^Af(u)K_{A^{-1}}(t,u){\rm{d}}u,
\]
where
\[
C=e^{\frac{j}{2}\left(cd{y_0}^2-2ady_0\omega_0+ab{\omega_0}^2\right)}.
\]

Let $f\in L^2(\mathbb{R})$, the continuous FT of $f$ is defined by
\begin{equation}\label{ft}
\mathcal{F}f(u)=\int_{-\infty}^{+\infty}f(t)e^{-jut}{\rm{d}}t,
\end{equation}
and inverse FT is given by
\begin{equation}\label{ift}
f(t)=\mathcal{F}^{-1}[\mathcal{F}f(u)](t)=\frac{1}{2\pi}\int_{-\infty}^{+\infty}\mathcal{F}f(u)e^{jut}{\rm{d}}u.
\end{equation}
Then, the relationship between the continuous OLCT and FT can be described as follows:
\begin{align}
O^{A}f(u)=\frac{1}{\sqrt{j2\pi b}}e^{j\left[\frac{d}{2b}(u-y_0)^2+u\omega_0\right]}\mathcal{F}\left(f(\cdot)e^{j\frac{a}{2b}(\cdot)^2}\right)
   \left(\frac{u-y_0}{b}\right).\label{rel:COLCT}
\end{align}

The convolution of two functions $f,\;g\in L^2(\mathbb{R})$ is defined by:
\[
(f*g)(\tau)=\int_{-\infty}^{+\infty}f(t)g(\tau-t){\rm{d}}t.
\]

\subsection{Discrete OLCT}

In this paper, we consider a discrete-time signal $\mathbf{x}=(x[n])_{n\in \mathbb{Z}}\in \ell^2$ with support length $N$, i.e., there exists a constant $n_0\in{\mathbb{Z}}$, satisfies
\[
x[n]=0,\;\;\mbox{for}\; n<n_0\; \mbox{and}\; n\geq n_0+N.
\]
Similarly, the discrete Fourier transform of $\mathbf{x}$ is given by
\[
\mathcal{F}[\mathbf{x}](u)=\widehat{\mathbf{x}}(u)=\sum_{n\in \mathbb{Z}}x[n]\ e^{-jun}, \;\;u\in[-\pi,\pi).
\]
The autocorrelation signal $\mathbf{a}:=(a[n])_{n\in \mathbb{Z}}$ associated with $\mathbf{x}$ is denoted as
\[
a[n]:=\sum_{k\in\mathbb{Z}}{\overline{{x[k]}}}x[k+n], \;n\in \mathbb{Z}.
\]
Note that $\mathbf{a}$ is conjugate symmetric, that is to say, $a[n]=\overline{a[-n]}$. It is easy to prove
\[
A(u):=|\widehat{\mathbf{x}}(u)|^2 =\sum_{n\in \mathbb{Z}}\sum_{k\in \mathbb{Z}}x[n]\overline{x[k]}e^{-ju(n-k)}=\sum_{n\in \mathbb{Z}}a[n]e^{-jun}=\widehat{\mathbf{a}}(u),
\]
where $A(u)$ is defined as the nonnegative autocorrelation function of $\mathbf{x}$.

Discretizing the continuous OLCT form in (\ref{def:COLCT}), we can obtain the definition of discrete OLCT as follows.
\begin{Definition}
For real parameters $A=\left[
      \begin{array}{cc|c}
        a & b & y_0 \\
        c & d & \omega_0 \\
      \end{array}
      \right]$
with $ad-bc=1$, the OLCT of a discrete-time signal $\mathbf{x}=(x[n])_{n\in\mathbb{Z}}\in \ell^2$ is defined by
\begin{equation*}
O^{A}\mathbf{x}(u)=
    \begin{cases}
   \sum_{n\in\mathbb{Z}}x[n]h_{A}(n,u), & b\ne0,\\
     \sqrt{d}e^{j\frac{cd}{2}(u-y_0)^2+ju\omega_0}x[d(u-y_0)], & b=0,\\
     \end{cases}
\end{equation*}
where
\[
h_{A}(n,u)=\sqrt{\frac{1}{j2\pi b}}e^{\frac{j}{2b}\left[a n^{2}+d{y_0}^2+2n(y_{0}-u)-2u(d y_{0}-b\omega_{0})+d u^{2}\right]}.
\]
\end{Definition}
Similarly, we only consider the case of $b\neq 0$ in the rest of paper. The connection between discrete OLCT and discrete FT can be represented as
\begin{equation}\label{dis:OF}
O^{A}\mathbf{x}(u)=\sqrt{\frac{1}{j2\pi b }}e^{j\left[\frac{d}{2b}(y_{0}-u)^{2}+u\omega_{0}\right]}\mathcal{F}{\left[\mathbf{x}[\cdot]e^{\frac{ja}{2b}(\cdot)^{2}}\right]}\left(\frac{u-y_{0}}{b}\right).
\end{equation}

\subsection{Short-time OLCT}

Given a window function $\varphi\in L^2(\mathbb{R})$, the STFT of $f\in L^2(\mathbb{R})$, is given by
\begin{align}
\mathcal{V}_\varphi f(v,u)=\int_{-\infty}^{+\infty}f(t)\overline{\varphi (t-v)}e^{-jut}{\rm{d}}t,\;\; v,u\in \mathbb{R}.
\end{align}
Set $T_\tau f(t):=f(t-\tau)$, then we know
\begin{equation}\label{sft}
\mathcal{V}_\varphi f(v,u)=e^{-jvu}\mathcal{F}^{-1}\left[\mathcal{F}f \cdot T_u\overline{\mathcal{F}\varphi}\right](v).
\end{equation}

Similarly, the STOLCT of a function $f\in L^2(\mathbb{R})$ with a window function $\varphi\in L^2(\mathbb{R})$ is denoted as \cite{Huo19M}:
\begin{align}
\mathcal{V}_\varphi^{A} f(v,u)=\mathcal{V}_\varphi^{A} \left[f(t)\right](v,u)=\int_{-\infty}^{+\infty}f(t) \overline{\varphi (t-v)}K_{A}(t,u){\rm{d}}t, \;\; v,u\in \mathbb{R},
\end{align}
where the transform kernel $K_A$ is defined the same as in (\ref{ker}), parameters $a$,\;$b\neq 0$,\;$c,\;d$,\;$y_0$,\;$\omega_0\in\mathbb{R}$, and $ad-bc=1$.
It is easy to see that
\begin{equation}\label{re:St}
\mathcal{V}_\varphi^{A} f(v,u)=\frac{1}{\sqrt{j2\pi b}}e^{j\left[\frac{d}{2b}(u-y_0)^2+u\omega_0\right]}\mathcal{V}_\varphi\left\{f(\cdot)e^{\frac{ja(\cdot)^2}{2b}}\right\} \left(v,\frac{u-y_0}{b}\right).
\end{equation}
\subsection{Ambiguity function}

The ambiguity function of $f(t)\in L^2(\mathbb{R})$ is defined as \cite{Nurkifayah2024}
\begin{equation}\label{am}
  \mathcal{A}f(\tau,\eta)=\int_{-\infty}^{+\infty}f\left(t+\frac{\tau}{2}\right)\overline{f \left(t-\frac{\tau}{2}\right)}e^{-jt \eta}{\rm{d}}t,
\end{equation}
and the cross-ambiguity function of $f(t),g(t)\in L^2(\mathbb{R})$ is given by
\begin{equation}\label{cro-am}
  \mathcal{A}\left(f,g\right)(\tau,\eta)=\int_{-\infty}^{+\infty}f\left(t+\frac{\tau}{2}\right)\overline{g \left(t-\frac{\tau}{2}\right)}e^{-jt \eta}{\rm{d}}t.
\end{equation}
It is easy to see that
\begin{equation}\label{Add:S5}
\mathcal{A}f(\tau,\eta)=e^{\frac{j\tau\eta}{2}}\mathcal{V}_f f(\tau,\eta),
\end{equation}
and
\begin{equation}\label{Add:S6}
\mathcal{A}f(0,\eta)=\mathcal{F}\left(|f|^2\right)(\eta).
\end{equation}
It follows from \cite{Philippe2014} that
\begin{equation}\label{eq:am}
\mathcal{A}f=\mathcal{A}g\; \Leftrightarrow \;{\mbox{there exists a constant}}\,\lambda\in \mathbb{C}\; {\mbox{with}}\; |\lambda|=1, \,{\mbox{such that}}\,\, g=\lambda f.
\end{equation}

\subsection{Bandlimited function}
For $\Omega>0$, the Paley-Wiener space \cite{Wiener1934} of a bandlimited function $f\in L^2 (\mathbb{R})$ is given by
\[
PW_\Omega^2:= \left\{f\in L^2 (\mathbb{R}):\; {\rm{supp}}(\mathcal{F}f)\subseteq[-\Omega,\Omega]\right\}.
\]
In the following, we present some results on bandlimited functions in Paley-Wiener space, which play important roles in later proof of our main theorems.
\begin{Proposition}\cite{Shannon1949}\label{pro:sha}
    Let $f\in PW_\Omega^2$ with $\Omega>0$. Then,
    \[
    f(t)=\sum_{n\in\mathbb{Z}}f\left({\frac{n}{2\Omega}}\right){\frac{\sin(2\pi \Omega t-n\pi)}{2\pi \Omega t-n\pi}},\;\; t\in\mathbb{R}.
    \]
\end{Proposition}

\begin{Lemma}\cite[Lemma~1.6]{AlaifariRima2021}\label{lem:am}
Let $f\in PW_\Omega^2$ with $\Omega>0$. Then, $\mathcal{A}f$ is uniformly continuous, and ${\rm{supp}}(\mathcal{A}f)\subset \mathbb{R}\times[-2\Omega,2\Omega]$.
\end{Lemma}

\section{Phase retrieval from magnitudes of the continuous OLCT}\label{CPL}

In this section, we consider the OLCT phase retrieval problem of continuous function $f(t)\in L^2(\mathbb{R})$. Similar to FT phase retrieval, the OLCT phase retrieval cannot be solved uniquely. Therefore, we first consider the non-trivial ambiguities in the OLCT phase retrieval. In addition, we consider the uniqueness of phase retrieval from multiple OLCTs.

\subsection{Ambiguities in the continuous OLCT phase retrieval}

In this subsection, we consider the OLCT phase retrieval problem, that is to say, our goal is to reconstruct a function $f(t)\in{L^2(\mathbb{R})}$ with compact support from the magnitude of $O^Af$. We first investigate three trivial ambiguities which cannot be avoided for uniqueness of OLCT phase retrieval.

\begin{Lemma}\label{lem:con}
Assume that $f(t)\in L^2(\mathbb{R})$ is a function with compact support. Then, the following three function
\begin{enumerate}[label=(\roman*)]
    \item  the rotated function $e^{j\beta}f$ with $\beta\in\mathbb{R}$;
    \item  the shifted function $e^{\frac{-jat_0(\cdot)}{b}}f(\cdot-t_0)$ with $t_0\in \mathbb{R}$;
    \item  the conjugate reflected function $e^{\frac{-ja{(\cdot)}^2}{b}}\overline{f(-\cdot)}$;
\end{enumerate}
have the same magnitude of OLCT measurements $|O^Af|$.
\end{Lemma}

\begin{proof}
By (\ref{rel:COLCT}) and the properties of the FT, we can easily obtain:

\begin{enumerate}[label=(\roman*)]
\item  $O^A\left[e^{j\beta}f\right](u)=e^{j\beta}O^Af(u)$;
\item \begin{align*}
      &O^A\left[e^{\frac{-jat_0(\cdot)}{b}}f(\cdot-t_0)\right](u)\\
      &=\sqrt{\frac{1}{j2\pi b}} e^{j\left[\frac{d}{2b}(u-y_0)^2+u\omega_0\right]}e^\frac{-ja{t_0}^2}{2b}\mathcal{F}\left[f(\cdot-t_0)e^{\frac{j}{2b}a(\cdot-t_0)^2}\right]\left(\frac{u-y_0}{b}\right)\\
      &=e^{\frac{-ja{t_0}^2}{2b}}e^{-jt_0\frac{u-y_0}{b}}O^Af(u);
\end{align*}
\item  \begin{align*}
O^A\left[e^{\frac{-ja{(\cdot)}^2}{b}}\overline{f(-\cdot)}\right](u)&=\sqrt{\frac{1}{j2\pi b}} e^{j\left[\frac{d}{2b}(u-y_0)^2+u\omega_0\right]}\mathcal{F}\overline{\left[{f(-\cdot)e^\frac{ja(-\cdot)^2}{2b}}\right]}\left(\frac{u-y_0}{b}\right)\\
&=je^{j\left[\frac{d}{b}(u-y_0)^2+2uw_0\right]}\overline{O^Af(u)}.
\end{align*}
\end{enumerate}
This completes the proof.
\end{proof}

Since the operations of global phase change, time shift, and conjugate reflect on the signal cannot be avoided, recovery is possible only up to these
trivial ambiguities. Next, we consider the nontrivial ambiguities in OLCT phase retrieval.

\begin{Theorem}\label{Thm:Na}
Suppose that $f_1,f_2\in L^2({\mathbb{R}})$ are two functions with compact support. Then, $|O^Af_1|=|O^Af_2|$ if and only if there exist two functions $g_1,g_2\in L^2(\mathbb{R})$ with compact support, satisfy
\[
f_1=e^{\frac{-ja{(\cdot)}^{2}}{2b}}\left(g_{1}*g_{2}\right),
\]
and
\[
f_2=e^{j\beta}e^{\frac{-ja{(\cdot)}^{2}}{2b}}\left(g_{1}(\cdot-t_{0})*\overline{{{g_{2}(-\cdot})}}\right)
\]
for some $\beta\in[-\pi,\pi)$, $t_0\in\mathbb{R}$. Here, $*$ represents a continuous convolution operator.
\end{Theorem}

\begin{proof}
$\Leftarrow$: Assume that $g_1,g_2\in L^2(\mathbb{R})$ are two functions with compact support, satisfy
\[
f_1=e^{\frac{-ja{(\cdot)}^{2}}{2b}}\left(g_{1}*g_{2}\right),
\]
and
\[
f_2=e^{j\beta}e^{\frac{-ja{(\cdot)}^{2}}{2b}}\left(g_{1}(\cdot-t_{0})*\overline{{{g_{2}(-\cdot})}}\right)
\]
for some $\beta \in [-\pi,\pi)$, $t_0 \in \mathbb{R}$. After simple calculations, we obtain
\[
|O^Af_1(u)|=|O^Af_2(u)|,
\]
for all $u\in \mathbb{R}$.

$\Rightarrow$: Suppose that $\left|O^Af_1(u)\right|=\left|O^Af_2(u)\right|$ for all $u\in\mathbb{R}$. Set
\[
\widetilde{f_1}:=f_1e^{\frac{ja{(\cdot)}^2}{2b}}\,\,and\,\,\widetilde{f_2}:=f_2e^{\frac{ja{(\cdot)}^2}{2b}}.
\]
From (\ref{rel:COLCT}), we get
\[
\mathcal{F}\widetilde{f_1}(u)\overline{\mathcal{F}\widetilde{f_1}(\overline{u})}=\left|\mathcal{F}\widetilde{f_1}(u)\right|^2
=\left|\mathcal{F}\widetilde{f_2}(u)\right|^2=\mathcal{F}\widetilde{f_2}(u)\overline{\mathcal{F}\widetilde{f_2}(\overline{u})},
\]
for all $u\in \mathbb{R}$. Following the proof process of \cite[Theorem~3.2]{ChenYang2022}, we can show
\[
f_1=e^{\frac{-ja{(\cdot)}^{2}}{2b}}\left(g_{1}*g_{2}\right),\,{\mbox{and}}\; f_2=e^{j\beta}e^{\frac{-ja{(\cdot)}^{2}}{2b}}\left(g_{1}(\cdot-t_{0})*\overline{{{g_{2}(-\cdot})}}\right),
\]
which completes the proof.
\end{proof}

\subsection{Phase retrieval from multiple OLCTs}

To obtain more information about the signal, we can solve the OLCT phase retrieval by adding extra magnitude-only OLCT measurements. In this subsection,
we show that an original signal can be uniquely recovered by the magnitude of multiple OLCTs up to a global phase.

First, we derive the relationship between the ambiguities functions of $f,\;g$ and $O^Af,\;O^Ag$.
\begin{Lemma}\label{lem:relat}
Let $f(t),g(t)\in L^2(\mathbb{R})$, then
\begin{equation}\label{relat}
\mathcal{A}\left(O^Af,O^Ag\right)(\tau,\eta)=e^{j\left[\tau \omega_0 -\eta y_0+\frac{1}{2}(d\tau-b\eta)(a\eta-c\tau)\right]}\mathcal{A}(f,g)(d\tau-b\eta,a\eta-c\tau).
\end{equation}
\end{Lemma}

\begin{proof}
By (\ref{rel:COLCT}), we obtain
\begin{align*}
&\mathcal{A}\left(O^Af,O^Ag\right)(\tau,\eta)\\
&=\int_{-\infty}^{+\infty}O^Af\left(t+\frac{\tau}{2}\right)\overline{O^Ag \left(t-\frac{\tau}{2}\right)}e^{-jt \eta}{\rm{d}}t\\
&=e^{\frac{j \tau \eta}{2}}\int_{-\infty}^{+\infty}O^Af(t)\overline{O^Ag(t-\tau)}e^{-jt \eta}{\rm{d}}t\\
&=\mu \int_{-\infty}^{+\infty}\mathcal{F}\left(f(\cdot)e^{j\frac{a}{2b}(\cdot)^2}\right)\left(\frac{t-y_0}{b}\right)
  \overline{\mathcal{F}\left(g(\cdot)e^{j\frac{a}{2b}(\cdot)^2}\right)\left(\frac{t-y_0-\tau}{b}\right)}e^{jt\left(\frac{d\tau}{b}-\eta\right)}{\rm{d}}t\\
&=b\mu e^{\frac{jy_0}{b}(d\tau-b\eta)}\int_{-\infty}^{+\infty}\mathcal{F}\left(f(\cdot)e^{j\frac{a}{2b}(\cdot)^2}\right)(t)
\overline{\mathcal{F}\left(g(\cdot)e^{j\frac{a}{2b}(\cdot)^2}\right)\left(t-\frac{\tau}{b}\right)}e^{jt(d\tau-b\eta)}{\rm{d}}t\\
&=2\pi b\mu e^{\frac{jy_0}{b}(d\tau-b\eta)}\int_{-\infty}^{+\infty}f(t)e^{j\frac{a}{2b}t^2}
\overline{g(t+b\eta-d\tau)e^{j\frac{a}{2b}(t+b\eta-d\tau)^2}}e^{j\frac{\tau}{b}(d\tau-b\eta-t)}{\rm{d}}t\\
&=e^{j\left[\tau \omega_0 -\eta y_0+\frac{1}{2}(d\tau-b\eta)(a\eta-c\tau)\right]}\int_{-\infty}^{+\infty}f(t)\overline{g(t+b\eta-d\tau)}e^{-jt(a\eta-c\tau)}{\rm{d}}t\\
&=e^{j\left[\tau \omega_0 -\eta y_0+\frac{1}{2}(d\tau-b\eta)(a\eta-c\tau)\right]}\mathcal{A}(f,g)(d\tau-b\eta,a\eta-c\tau),
\end{align*}
where
\[
\mu=\frac{1}{2\pi b}e^{j\left[\frac{\tau\eta}{2}-\frac{d}{2b}(\tau ^2+2y_0 \tau)+\tau{\omega_0}\right]}.
\]
\end{proof}

Specially, let $f=g$, and $\tau=0$ in Lemma ~\ref{lem:relat}, we have
\begin{equation}\label{eq:relat1}
e^{-j\eta \left(y_0+\frac{1}{2}ab\eta\right)}\mathcal{A}f(-b\eta,a\eta)
=\mathcal{A}\left(O^Af\right)(0,\eta)=\mathcal{F}\left[\left|O^Af\right|^2\right](\eta).
\end{equation}

Next, we show that a compactly supported function can be uniquely recovered by its multiple OLCT magnitude measurements.
\begin{Theorem}
Suppose that two functions $f,g\in L^2(\mathbb{R})$ have compact support, $\Lambda \subset \mathbb{R}$ is a set with positive measure, and
\[
\mathbf{\Lambda} =\left\{\frac{a}{b} \in \Lambda,\, b>0, \,ad-bc=1, {\rm{and}}\; a,b,c,d,y_0,\omega_0\in \mathbb{R} \right\}.
\]
If $\left|{O}^{A}f\right|=\left|{O}^{A}g\right|$ holds for all $A\in \mathbf{\Lambda}$, then there exists a constant $\lambda\in \mathbb{C}$ with $|\lambda|=1$, such that $g=\lambda f$.
\end{Theorem}
\begin{proof}
For any $\tau\in \mathbb{R}$, set
\[
\mathbf{\Lambda}_\tau :=\left\{{-\frac{a}{b}\tau },A\in \mathbf{\Lambda}\right\}.
\]
Then, $\mathbf{\Lambda}_\tau$ is a set with positive measure. By (\ref{eq:relat1}), we get
\[
\mathcal{A}f(-b \eta,a\eta)=\mathcal{A}g(-b\eta,a\eta)
\]
for all $\eta \in \mathbb{R}$ and $A\in \mathbf{\Lambda}$.
Hence, for all $\sigma\in {\mathbf{\Lambda}_\tau}$, we obtain
\begin{equation*}
\mathcal{A}f(\tau,\sigma)=\mathcal{A}g(\tau,\sigma).
\end{equation*}
By Paley-Wiener theorem \cite{Louis1968}, we can extend the ambiguity functions of $f,g\in L^2(\mathbb{R})$ to entire functions with respect to the second variable $\sigma$. Therefore, $\mathcal{A}f=\mathcal{A}g$. By (\ref{eq:am}), we know that there exists a constant $\lambda\in \mathbb{C}$ with $|\lambda|=1$, such that $g=\lambda f$, which completes the proof.
\end{proof}

\section{The discrete OLCT phase retrieval}\label{DPR}

In this section, we consider the discrete OLCT phase retrieval problem: Whether a discrete signal $\mathbf{x}={(x[n])_{n\in\mathbb{Z}}\in \ell^2}$ with finite support can be uniquely reconstructed from its magnitude of OLCT measurements.

At the beginning, we summarize three trivial ambiguities, which are caused by rotation, translation, and conjugate reflection.

\begin{Lemma}\label{lem:dis}
Assume that $\mathbf{x}\in {\ell^{2}(\mathbb{Z})}$ is a signal with finite support. Then, the following three functions
\begin{enumerate}[label=(\roman*)]
    \item the rotated function $e^{j\beta}\mathbf{x},\beta\in\mathbb{R}$;
    \item the shifted function $e^{\frac{-jan_{0}(\cdot)}{b}}\mathbf{x}[\cdot-n_{0}],n_{0}\in \mathbb{Z}$;
    \item the conjugate reflected function $e^{\frac{-ja{(\cdot)}^{2}}{b}}{\overline{{\mathbf{x}[-\cdot]}}}$;
\end{enumerate}
have the same magnitude of OLCT measurements $|O^A\mathbf{x}|$.
\end{Lemma}
The proof of Lemma \ref{lem:dis} is similar to Lemma \ref{lem:con}. In order not to repeat, we leave out all the details.

In the following, we investigate nontrivial ambiguities of discrete OLCT phase retrieval.

\begin{Theorem}\label{thm:dis}
Let $\mathbf{x},\,\mathbf{y}\in \ell^2(\mathbb{Z})$ be finite complex signals. Then, $\mathbf{x}$ and $\mathbf{y}$ have same magnitudes of OLCT if and only if there exist finite signals $\mathbf{x}_1,\mathbf{x}_2 \in \ell^2(\mathbb{Z})$, such that
\begin{equation}\label{thm:dis1}
      \mathbf{x}=e^{\frac{-ja{(\cdot)}^2}{2b}}({\mathbf{x}_1}*{\mathbf{x}_2})\,\,{\mbox{and}}\,\, \mathbf{y}=e^{j\beta}e^{\frac{-ja{(\cdot)}^2}{2b}}\left({\mathbf{x}_1[\cdot -n_0]*\overline{\mathbf{x}_2[-\cdot]}}\right)
\end{equation}
for some $\beta\in [-\pi,\pi), \,n_0\in \mathbb{Z}$. Here, $*$ represents the discrete convolution operator.
\end{Theorem}

\begin{proof}
$\Leftarrow$: Let $\mathbf{x}_1,\mathbf{x}_2 \in \ell^2(\mathbb{Z})$ be finite complex signals, such that
\[
      \mathbf{x}=e^{\frac{-ja{(\cdot)}^2}{2b}}({\mathbf{x}_1}*{\mathbf{x}_2})\,\,{\mbox{and}}\,\,
      \mathbf{y}=e^{j\alpha}e^{\frac{-ja{(\cdot)}^2}{2b}}\left({\mathbf{x}_1[\cdot -n_0]*\overline{\mathbf{x}_2[-\cdot]}}\right),
\]
By (\ref{dis:OF}), we get
\[
|O^A\mathbf{x}(u)|=|O^A\mathbf{y}(u)|.
\]

$\Rightarrow$: Set
\[
\widetilde{\mathbf{x}}:=\mathbf{x}e^{\frac{ja{(\cdot)}^2}{2b}}\,\,{\mbox{and}}\,\,\widetilde{\mathbf{y}}:=\mathbf{y}e^{\frac{ja{(\cdot)}^2}{2b}}.
\]
From (\ref{dis:OF}) and $|O^A\mathbf{x}(u)|=|O^A\mathbf{y}(u)|$, we know that
\[
|\mathcal{F}[\widetilde{\mathbf{x}}](u)|=|\mathcal{F}[\widetilde{\mathbf{y}}](u)|.
\]
Let $\widetilde{A}(u)$ be the autocorrelation signal of $\widetilde{\mathbf{x}}$. Then, we have
\[
\widetilde{A}(u)=|\mathcal{F}[\widetilde{\mathbf{x}}](u)|^2=|\mathcal{F}[\widetilde{\mathbf{y}}](u)|^2.
\]
Following the same procedure for the proof of \cite[Theorem 2.3]{Beinert2015}, yields
\[
\mathcal{F}[\widetilde{\mathbf{x}}](u)=\mathcal{F}[\mathbf{x}_1](u)\mathcal{F}[\mathbf{x}_2](u),
\]
and
\[
\mathcal{F}[\widetilde{\mathbf{y}}](u)=e^{-jun_0}\mathcal{F}[\mathbf{x}_1](u)\overline{\mathcal{F}[\mathbf{x}_2](u)}.
\]
By the relationship between discrete convolution operator and FT, we can get (\ref{thm:dis1}), and the proof is completed.
\end{proof}

In Theorem \ref{thm:dis}, we prove that the nontrivial ambiguities in the discrete OLCT phase retrieval can be characterized by convolution operations.

\section{The STOLCT phase retrieval}\label{STPR}

In this section, we focus on introducing redundancy into the OLCT magnitude-only measurements by utilizing the STOLCT. The redundancy is obtained by a significant overlap between adjacent short-time sections. We show that the redundancy provided by the STOLCT can guarantee the unique recovery of nonseparable (or bandlimited) signals from the STOLCT magnitude-only measurements.

At the begining, we investigate the relation between the STOLCT and ambiguity function.

\begin{Lemma}\label{lem:ST}
Let $f,\;\varphi\in L^2(\mathbb{R})$. Then, we have
\begin{align}
\mathcal{F}(|\mathcal{V}_\varphi^{A} f|^2)(u,-v)
&=e^{jy_0v}\mathcal{A}\tilde{f}(bv,u)\overline{\mathcal{A}\varphi(bv,u)}\nonumber\\
&=e^{jy_0 v}\mathcal{A}f(bv,u-av)\overline{\mathcal{A}\varphi(bv,u)},\;\; u,v\in \mathbb{R},\label{ST:1}
\end{align}
where $\tilde{f}(t):=f(t)e^{\frac{jat^2}{2b}}$.
\begin{proof}
Fix a $v\in \mathbb{R}$, let $f_v(t):=f(t)\overline{\varphi(t-v)}e^{\frac{jat^2}{2b}}$. It is easy to obtain that $f_v(t)\in L^1(\mathbb{R})$, and
\begin{align}
\mathcal{V}_\varphi^{A}f(v,u)&=\frac{1}{\sqrt{j2\pi b}}e^{j\left[\frac{d}{2b}(u-y_0)^2+u\omega_0\right]} \int_{-\infty}^{+\infty}f(t)\overline{\varphi (t-v)}e^{\frac{jat^2}{2b}}e^{-j\frac{u-y_0}{b}t}{\rm{d}}t\nonumber\\
&=\frac{1}{\sqrt{j2\pi b}}e^{j\left[\frac{d}{2b}(u-y_0)^2+u\omega_0\right]}\mathcal{F}\left(f_v\right)\left(\frac{u-y_0}{b}\right).\label{Add:S1}
\end{align}
Set
 $f_v^{\sharp}(t):=\overline{f_v(-t)}$. Then, we have $\overline{\mathcal{F}(f_v)(u)}=\mathcal{F}(f_v^{\sharp})(u)$,
and
\begin{align}
\overline{\mathcal{V}_\varphi^{A}f(v,u)}=\frac{1}{\sqrt{-j2\pi b}}e^{-j\left[\frac{d}{2b}(u-y_0)^2+u\omega_0\right]} \mathcal{F}\left(f_v^{\sharp}\right)\left(\frac{u-y_0}{b}\right).\label{Add:S2}
\end{align}
Combining (\ref{Add:S1}) and (\ref{Add:S2}), we obtain
\begin{eqnarray}
|\mathcal{V}_\varphi^{A}f(v,u)|^2
&=&\frac{1}{2\pi b}\mathcal{F}\left(f_v\right)\left(\frac{u-y_0}{b}\right)\cdot\mathcal{F}\left(f_v^{\sharp}\right)\left(\frac{u-y_0}{b}\right)\nonumber\\
&=&\frac{1}{2\pi b}\mathcal{F}\left(f_v*f_v^{\sharp}\right)\left(\frac{u-y_0}{b}\right),\label{Add:S3}
\end{eqnarray}
where we use convolution theorem in the lase step.
Applying the inverse FT on the both sides of (\ref{Add:S3}), yields
\begin{align}
    \mathcal{F}\left(\left|\mathcal{V}_\varphi^{A}f(v,\cdot)\right|^2 \right)(-v')&=e^{jy_0v^{\prime}}\left(f_v*f_v^{\sharp}\right)(bv^{\prime})\nonumber\\
    &=e^{jy_0v^{\prime}} \int_{-\infty}^{+\infty}f_v(t)f_v^{\sharp}(bv^{\prime}-t){\rm{d}}t\nonumber\\
    &=e^{jy_0v^{\prime}} \int_{-\infty}^{+\infty}\tilde{f}(t)\overline{\varphi(t-v)}\overline{\tilde{f}(t-bv^{\prime})}\varphi(t-bv^{\prime}-v){\rm{d}}t,\label{Add:S4}
\end{align}
where $\tilde{f}(t):=f(t)e^{\frac{jat^2}{2b}}$ for $t\in \mathbb{R}$. By performing FT on both sides of (\ref{Add:S4}) to variable $v$, we have
\begin{align*}
    &\mathcal{F}\left(\left|\mathcal{V}_\varphi^{A}f\right|^2 \right)(u^{\prime},-v^\prime)\\
    =&e^{jy_0v^{\prime}}\int_{-\infty}^{+\infty}\int_{-\infty}^{+\infty}\tilde{f}(t)\overline{\varphi(t-v)}
     \overline{\tilde{f}(t-bv^{\prime})}\varphi(t-bv^{\prime}-v)e^{-jvu^{\prime}}{\rm{d}}t{\rm{d}}v \\
    =&e^{jy_0v^{\prime}}\int_{-\infty}^{+\infty}\tilde{f}(t)\overline{\tilde{f}(t-bv^{\prime})}e^{-jtu^{\prime}}
    \int_{-\infty}^{+\infty}\overline{\varphi(t-v)}\varphi(t-bv^{\prime}-v)e^{j(t-v)u^{\prime}}{\rm{d}}v{\rm{d}}t \\
    =&e^{jy_0v^{\prime}} \int_{-\infty}^{+\infty}\tilde{f}(t)\overline{\tilde{f}(t-bv^{\prime})}e^{-jtu^{\prime}}{\rm{d}}t
    \cdot\int_{-\infty}^{+\infty}\overline{\varphi(v)}\varphi(v-bv^{\prime})e^{jvu^{\prime}}{\rm{d}}v \\
    =&e^{jy_0v^{\prime}}\mathcal{V}_{\tilde{f}}\tilde{f}(bv^{\prime},u^{\prime})\cdot \overline{\mathcal{V}_{\varphi}\varphi(bv^{\prime},u^{\prime})}\\
    =&e^{jy_0v^{\prime}}\mathcal{A}\tilde{f}(bv^{\prime},u^{\prime})\overline{\mathcal{A}\varphi(bv^{\prime},u^{\prime})}.
\end{align*}
Furthermore,
\begin{align*}
    \mathcal{V}_{\tilde{f}}\tilde{f}(bv^{\prime},u^{\prime})
    =&e^{-\frac{jab{v^{\prime}}^2}{2}} \int_{-\infty}^{+\infty}f(t)\overline{f(t-bv^{\prime})}e^{-jt(u^{\prime}-av^{\prime})}{\rm{d}}t\\
    =&e^{-\frac{jab{v^{\prime}}^2}{2}} \mathcal{V}_{f}f(bv^{\prime},u^{\prime}-av^{\prime}) \\
    =&e^{-\frac{jbv^{\prime}u^{\prime}}{2}}\mathcal{A}f(bv^{\prime},u^{\prime}-av^{\prime}),
\end{align*}
which completes the proof.
\end{proof}
\end{Lemma}

\subsection{The STOLCT phase retrieval for nonseparable signals}

In this subsection, we derive some sufficient and necessary conditions that can guarantee a square integrable function can be uniquely recovered from the magnitude-only STOLCT measurements.

We call a function $f\in L^2(\mathbb{R})$ is separable if there exist two nonzero functions $f_1,f_2\in L^2(\mathbb{R})$, such that $f=f_1+f_2$ and $f_1f_2=0$. We first review a sufficient and necessary condition that can guarantee unique recovery of a signal from its magnitude-only measurements.

\begin{Proposition}\cite[Theorem~II.2]{YangChen2020}\label{pro:sep}
    Assume that $V$ is a linear space of real-valued continuous functions on the real line $\mathbb{R}$. Then, a function $f\in V$ can be uniquely determined by its magnitude-only measurements $\left\{|f(t)|, t\in \mathbb{R}\right\}$, up to a sign, if and only if $f$ is nonseparable.
\end{Proposition}

Now, we prove a uniqueness result on STOLCT phase retrieval for nonseparable real-valued continuous functions.

\begin{Theorem}\label{thm:non}
  Assume that a window function $\varphi \in L^2(\mathbb{R})$ satisfies
  \[
    \mathcal{A}\varphi(0,u)\neq 0,\,\;{\mbox{for}}\;\; a.e.\;\; u\in \mathbb{R},
  \]
  then, the following two statements are equivalent for two nonseparable real-valued continuous functions $f,\,g\in L^2(\mathbb{R})$:
  \begin{enumerate}[label=(\roman*)]
\item  $f =\pm g$.
\item  $|\mathcal{V}_\varphi^{A} f|=|\mathcal{V}_\varphi^{A} g|$.
\end{enumerate}
\begin{proof}
$(i)\Rightarrow(ii)$: If $f =\pm g$, it can be easily obtained that
\[
|\mathcal{V}_\varphi^{A} f|=|\mathcal{V}_\varphi^{A} g|.
\]

$(ii)\Rightarrow(i)$: Assume that $|\mathcal{V}_\varphi^{A} f|=|\mathcal{V}_\varphi^{A} g|$. From Lemma~\ref{lem:ST}, we know
\[
\mathcal{A}f(bv,u-av)\overline{\mathcal{A}\varphi(bv,u)}=\mathcal{A}g(bv,u-av)\overline{\mathcal{A}\varphi(bv,u)},\;\;{\mbox{for}}\;\; u,v\in\mathbb{R}.
\]
Since $\mathcal{A}\varphi(0,u)\neq 0$, for a.e. $u\in \mathbb{R}$, we get
\[
\mathcal{A}f(0,u)=\mathcal{A}g(0,u),\,\;{\mbox{for}}\;\; a.e.\;\; u\in \mathbb{R}.
\]
Since $f,\,g\in L^2(\mathbb{R})$, by (\ref{Add:S6}), we know that
$\mathcal{A}f(0,u)$ and $\mathcal{A}g(0,u)$ are continuous. Hence,
\[
\mathcal{A}f(0,u)=\mathcal{A}g(0,u),\,\;{\mbox{for all}}\;\, u\in \mathbb{R},
\]
i.e., $|f|=|g|$. It follows from Proposition~\ref{pro:sep} that $f=\pm g$. This completes the proof.
\end{proof}
\end{Theorem}

Next, we state that a function $f\in L^2(\mathbb{R})$ can be uniquely determined from its magnitude-only STOLCT measurements $|\mathcal{V}_\varphi^{A} f|$ up to a global phase, providing the ambiguity function of window function $\varphi \in L^2(\mathbb{R})$ is nowhere vanishing.

\begin{Theorem}\label{thm:l2}
  Suppose that a window function $\varphi \in L^2(\mathbb{R})$ satisfies
  \[
    \mathcal{A}\varphi(v,u)\neq 0,\,\;{\mbox{for}}\;\; a.e.\;\; (v,u)\in \mathbb{R}^2,
  \]
  then, the following assertions are equivalent for two functions $f,\,g\in L^2(\mathbb{R})$:
  \begin{enumerate}[label=(\roman*)]
\item  $f =e^{j\beta} g$, for some $\beta\in \mathbb{R}$.
\item  $|\mathcal{V}_\varphi^{A} f|=|\mathcal{V}_\varphi^{A} g|$.
\end{enumerate}
\end{Theorem}

\begin{proof}
$(i)\Rightarrow(ii)$: Assume that $f =e^{j\beta}g$. Obviously, we have
\[
|\mathcal{V}_\varphi^{A} f|=|\mathcal{V}_\varphi^{A} g|.
\]

$(ii)\Rightarrow(i)$: Assume that $|\mathcal{V}_\varphi^{A} f|=|\mathcal{V}_\varphi^{A} g|$. From Lemma~\ref{lem:ST}, we know
\begin{equation}\label{A1}
\mathcal{A}f(bv,u-av)\overline{\mathcal{A}\varphi(bv,u)}=\mathcal{A}g(bv,u-av)\overline{\mathcal{A}\varphi(bv,u)},\;\;{\mbox{for}}\;\; (u,v)\in\mathbb{R}^2.
\end{equation}
Since $\mathcal{A}\varphi(v,u)\neq 0$, for a.e. $(v,u)\in \mathbb{R}^2$, we get
\[
\mathcal{A}f(bv,u-av)=\mathcal{A}g(bv,u-av),\;\;{\mbox{for a.e. }}\;\; (u,v)\in\mathbb{R}^2.
\]
According to a similar procedure for the proof of \cite[Theorem 3.1]{LiRui2024}, we know
\[
\langle f, \Phi\rangle=e^{j\beta}\langle g, \Phi\rangle,\;\; {\mbox{for any}}\;\; \Phi\in L^2(\mathbb{R}).
\]
Therefore, $f =e^{j\beta} g$.
\end{proof}

\subsection{The sampled STOLCT phase retrieval for bandlimted signals}
In this subsection, we investigate the uniqueness results for a bandlimited function which is recovered from its sampled STLOCT measurements. First, we consider a class of functions that are bandlimited in classical FT domain.
\begin{Theorem}\label{Thm:Sam}
For $\Omega_{1},\;\Omega_{2}>0$, let $\varphi \in PW_{\Omega_{1}}^2$ satisfy
\[
\mathcal{A}\varphi(0,u)\neq0,\,\;{\mbox{for}}\;\; a.e.\;\; u\in(-2\Omega_2,2\Omega_2),
\]
and $\mathcal{F}{\varphi}$ be bounded.
Then, the following two statements are equivalent for two real-valued continuous functions $f,g\in PW_{\Omega_2}^2$ on the real line:
\begin{enumerate}[label=(\roman*)]
    \item  $f=\pm g$.
    \item  $\left|\mathcal{V}_\varphi^{A} f\left(\frac{n}{4\Omega_{b,y_0}^{u}},u\right)\right|=\left|\mathcal{V}_\varphi^{A} g\left(\frac{n}{4\Omega_{b,y_0}^u},u\right)\right|$, \;\;$n\in\mathbb{Z}$, \\\\
    where $\Omega_{b,y_0}^u=\max\left\{\left|\Omega_1+\frac{u-y_0}{b}\right|,\left|-\Omega_1
        +\frac{u-y_0}{b}\right|\right\}.$
\end{enumerate}
\begin{proof}
$(i)\Rightarrow(ii)$: Assume that $f =\pm g$. Then, we have
\begin{equation}\label{Add:S7}
\left|\mathcal{V}_\varphi^{A} f\left(\frac{n}{4\Omega_{b,y_0}^{u}},u\right)
\right|=\left|\mathcal{V}_\varphi^{A} g\left(\frac{n}{4\Omega_{b,y_0}^u},u\right)\right|, \;\;n\in\mathbb{Z}.
\end{equation}

$(ii)\Rightarrow(i)$: Suppose that (\ref{Add:S7}) holds. Let $\tilde{f}(t):=f(t)e^{\frac{jat^2}{2b}}$, and $T_\tau f(t):=f(t-\tau)$.
Combining (\ref{sft}) and (\ref{re:St}), we have
\begin{eqnarray}
|\mathcal{V}_\varphi^{A} f(v,u)|^2&=&\frac{1}{2\pi b}\mathcal{V}_\varphi \tilde{f}\left(v,\frac{u-y_0}{b}\right)\cdot\overline{\mathcal{V}_\varphi \tilde{f}\left(v,\frac{u-y_0}{b}\right)}\nonumber\\
&=&\frac{1}{2 \pi b }\mathcal{F}^{-1}\left[\mathcal{F}\tilde{f} \cdot T_{\frac{u-y_0}{b}}\overline{\mathcal{F}{\varphi}}\right](v) \cdot \overline{\mathcal{F}^{-1}\left[\mathcal{F}\tilde{f} \cdot T_{\frac{u-y_0}{b}}\overline{\mathcal{F}{\varphi}}\right](v)}.\label{sam:1}
\end{eqnarray}
Set $f^{\sharp}(t):=\overline{f(-t)}$. Then, we know $\mathcal{F}^{-1}(f^{\sharp})=\overline{\mathcal{F}^{-1}f}$.
It follows from convolution theorem that
\begin{equation}\label{Add:S8}
\mathcal{F}(|\mathcal{V}_\varphi^{A} f(\cdot,u)|^2)=\frac{1}{b}\left[\left(\mathcal{F}\tilde{f} \cdot T_{\frac{u-y_0}{b}}\overline{\mathcal{F}{\varphi}}\right)*\left(\mathcal{F}\tilde{f} \cdot T_{\frac{u-y_0}{b}}\overline{\mathcal{F}{\varphi}}\right)^{\sharp} \right],
\end{equation}
for fixed $u\in \mathbb{R}$.
Since ${\rm{supp}}(\mathcal{F}{\varphi})\subset [-\Omega_1,\Omega_1]$, we have
\[
{\rm{supp}}\left(T_{\frac{u-y_0}{b}}{\mathcal{F}{\varphi}}\right) \subset [-\Omega_1+\frac{u-y_0}{b},\Omega_1+\frac{u-y_0}{b}]\subset[-\Omega_{b,y_0}^u,\Omega_{b,y_0}^u].
\]
Hence,
\[
{\rm{supp}}\left(\mathcal{F}\tilde{f} \cdot T_{\frac{u-y_0}{b}}\overline{\mathcal{F}{\varphi}}\right)\subset[-\Omega_{b,y_0}^u,\Omega_{b,y_0}^u].
\]
By (\ref{Add:S8}), we know that
\[
{\rm{supp}}\left(\mathcal{F}|\mathcal{V}_\varphi^{A} f(\cdot,u)|^2\right)\subset[-2\Omega_{b,y_0}^u,2\Omega_{b,y_0}^u].
\]
Moreover, since $\mathcal{F}{\varphi}$ is bounded, $\mathcal{F}\tilde{f} \cdot T_{\frac{u-y_0}{b}}\overline{\mathcal{F}{\varphi}}\in L^2(\mathbb{R})$. Thus, $|\mathcal{V}_\varphi^{A} f(\cdot,u)|^2\in L^2(\mathbb{R})$. Therefore, $|\mathcal{V}_\varphi^{A} f(\cdot,u)|^2\in PW_{2\Omega_{b,y_0}^u}^2$.
By Proposition~\ref{pro:sha}, we obtain
   \[
   |\mathcal{V}_\varphi^{A} f(v,u)|=|\mathcal{V}_\varphi^{A} g(v,u)|,\,\;{\mbox{for any}}\;\; u \in \mathbb{R.}
   \]
Applying Lemma~\ref{lem:ST}, we get
\begin{equation}\label{A2}
\mathcal{A}f(bv,u-av)\overline{\mathcal{A}\varphi(bv,u)}=\mathcal{A}g(bv,u-av)\overline{\mathcal{A}\varphi(bv,u)},\;\;{\mbox{for}}\;\; (u,v)\in\mathbb{R}^2.
\end{equation}
Since
\[
\mathcal{A}\varphi(0,u)\neq0,\,\;{\mbox{for}}\;\; a.e.\;\; u\in(-2\Omega_2,2\Omega_2),
\]
we have
\[
\mathcal{A}f(0,u)=\mathcal{A}g(0,u),\;\;{\mbox{for}}\;\,a.e.\;\; u\in (-2\Omega_2,2\Omega_2).
\]
As $f,g\in PW_\Omega^2$, it follows from Lemma \ref{lem:am} that $\mathcal{A}f$ and $\mathcal{A}g$ are uniformly continuous, and ${\rm{supp}}(\mathcal{A}f),\;{\rm{supp}}(\mathcal{A}g)\subset\mathbb{R}\times(-2\Omega_2,2\Omega_2)$.
Hence,
\[
\mathcal{A}f(0,u)=\mathcal{A}g(0,u),\;\;{\mbox{for all}}\;\; u\in \mathbb{R}.
\]
Therefore, $|f|=|g|$. By Proposition \ref{pro:sep}, we get $f=\pm g$, which completes the proof.
\end{proof}
\end{Theorem}

Finally, we consider the uniqueness of STOLCT phase retrieval for a complex-valued function, which is bandlimited in OLCT domain, instead of FT domain.

For a function $f\in L^2 (\mathbb{R})$ which is bandlimited in OLCT domain, we introduce a generalized Paley-Wiener space as follows:
\[
GPW_\Omega^2:= \left\{f\in L^2 (\mathbb{R}):\; {\rm{supp}}(O^Af)\subseteq[-\Omega,\Omega]\right\}.
\]

Then, we study the uniqueness resluts for a complex-valued bandlimited function from its sampled OLCT measurements.
\begin{Theorem}\label{Thm:Osam}
For $\Omega>0$,\;$\gamma\in(0,\frac{b}{2\Omega^{\prime}}]$, let $\varphi \in L^2(\mathbb{R})$ satisfy
\[
\mathcal{A}\varphi(0,u)\neq0,\,\;{\mbox{and}}\;\;\mathcal{A}\varphi(\gamma,u)\neq0,\;\; {\mbox{for}}\;\; a.e.\;\; u\in\left(-2{\frac{\Omega^{\prime}}{b}},2{\frac{\Omega^{\prime}}{b}}\right),
\]
and $\mathcal{F}{\varphi}$ be bounded.
Then, the following two assertions are equivalent for two complex-valued continuous functions $f,g\in GPW_{\Omega}^2$:
\begin{enumerate}[label=(\roman*)]
    \item  $f=e^{j\beta} g$, for some $\beta\in \mathbb{R}$.
    \item  $\left|\mathcal{V}_\varphi^{A} f\left(\frac{bn}{4\Omega^{\prime}},u\right)\right|=\left|\mathcal{V}_\varphi^{A} g\left(\frac{bn}{4\Omega^{\prime}},u\right)\right|$, \;\;$n\in\mathbb{Z}$,\;\; $\Omega^{\prime}=\max\{|\Omega-y_0|,|\Omega+y_0|\}$.
\end{enumerate}
\begin{proof}
$(i)\Rightarrow(ii)$: Assume that $f=e^{j\beta}g$. Obviously, we have
\begin{equation}\label{Add:S10}
\left|\mathcal{V}_\varphi^{A} f\left(\frac{bn}{4\Omega^{\prime}},u\right)
\right|=\left|\mathcal{V}_\varphi^{A} g\left(\frac{bn}{4\Omega^{\prime}},u\right)\right|, \;\;n\in\mathbb{Z}.
\end{equation}

$(ii)\Rightarrow(i)$: Suppose that (\ref{Add:S10}) holds. Using the same method as Theorem~\ref{Thm:Sam}, we can obtain
\begin{equation}\label{Add:S11}
\mathcal{F}(|\mathcal{V}_\varphi^{A} f(\cdot,u)|^2)=\frac{1}{b}\left[\left(\mathcal{F}\tilde{f} \cdot T_{\frac{u-y_0}{b}}\overline{\mathcal{F}{\varphi}}\right)*\left(\mathcal{F}\tilde{f} \cdot T_{\frac{u-y_0}{b}}\overline{\mathcal{F}{\varphi}}\right)^{\sharp} \right],
\end{equation}
for fixed $u\in \mathbb{R}$. Since $f\in GPW_{\Omega}^2$, by (\ref{rel:COLCT}), we get
\[
{\rm{supp}}(\mathcal{F}\tilde{f})\subset\left[-\frac{\Omega+y_0}{b},\frac{\Omega-y_0}{b}\right].
\]
Hence,
\[
{\rm{supp}}\left(\mathcal{F}\tilde{f} \cdot T_{\frac{u-y_0}{b}}\overline{\mathcal{F}{\varphi}}\right)
 \subset\left[-\frac{\Omega+y_0}{b},\frac{\Omega-y_0}{b}\right]\subset\left[-\frac{\Omega^{\prime}}{b},\frac{\Omega^{\prime}}{b}\right],
\]
where $\Omega^{\prime}=\max\{|\Omega-y_0|,|\Omega+y_0|\}$.
By (\ref{Add:S11}), we have
\[
{\rm{supp}}\left(\mathcal{F}|\mathcal{V}_\varphi^{A} f(\cdot,u)|^2\right)\subset\left[-\frac{2\Omega^{\prime}}{b},\frac{2\Omega^{\prime}}{b}\right],\,\;{\mbox{for fixed}}\;\; u\in \mathbb{R}.
\]
Furthermore, as $\mathcal{F}{\varphi}$ is bounded, $\mathcal{F}\tilde{f} \cdot T_{\frac{u-y_0}{b}}\overline{\mathcal{F}{\varphi}}\in L^2(\mathbb{R})$. Hence, $|\mathcal{V}_\varphi^{A} f(\cdot,u)|^2\in L^2(\mathbb{R})$. Therefore, $|\mathcal{V}_\varphi^{A} f(\cdot,u)|^2\in PW_{\frac{2\Omega^{\prime}}{b}}^2$.
By Proposition~\ref{pro:sha}, we obtain
   \[
   |\mathcal{V}_\varphi^{A} f(v,u)|=|\mathcal{V}_\varphi^{A} g(v,u)|,\,\;{\mbox{for all}}\;\; u \in \mathbb{R.}
   \]
It follows from Lemma~\ref{lem:ST} that
\begin{equation*}
\mathcal{A}\tilde{f}(bv,u)\overline{\mathcal{A}\varphi(bv,u)}=\mathcal{A}\tilde{g}(bv,u)\overline{\mathcal{A}\varphi(bv,u)},\;\;{\mbox{for}}\;\; (u,v)\in\mathbb{R}^2.
\end{equation*}
Since
\[
\mathcal{A}\varphi(0,u)\neq0,\,\;{\mbox{and}}\;\;\mathcal{A}\varphi(\gamma,u)\neq0,\;\;{\mbox{for}}\;\; a.e.\;\; u\in\left(-2{\frac{\Omega^{\prime}}{b}},2{\frac{\Omega^{\prime}}{b}}\right),
\]
we get
\[
\mathcal{A}\tilde{f}(0,u)=\mathcal{A}\tilde{g}(0,u),\;\;{\mbox{and}}\;\;\mathcal{A}\tilde{f}(\gamma,u)=\mathcal{A}\tilde{g}(\gamma,u),\;\;{\mbox{for}}\;\;\,a.e.\;\; u\in\left(-2{\frac{\Omega^{\prime}}{b}},2{\frac{\Omega^{\prime}}{b}}\right).
\]
As $f,g\in GPW_\Omega^2$, we know $\tilde{f},\tilde{g}\in PW_{\frac{\Omega^{\prime}}{b}}^2$.
Applying Lemma \ref{lem:am}, yields $\mathcal{A}\tilde{f}$ and $\mathcal{A}\tilde{g}$ are uniformly continuous, and ${\rm{supp}}(\mathcal{A}\tilde{f}),\;{\rm{supp}}(\mathcal{A}\tilde{g})\subset\mathbb{R}\times\left(-2{\frac{\Omega^{\prime}}{b}},2{\frac{\Omega^{\prime}}{b}}\right)$.
Thus
\[
\mathcal{A}\tilde{f}(0,u)=\mathcal{A}\tilde{g}(0,u),\;\;{\mbox{and}}\;\;\mathcal{A}\tilde{f}(\gamma,u)=\mathcal{A}\tilde{g}(\gamma,u),\;\;{\mbox{for all}}\;\; u\in \mathbb{R}.
\]
Therefore,
\[
|\tilde{f}|=|\tilde{g}|\;\;{\mbox{and}}\;\; \tilde{f}(t)\overline{\tilde{f}(t-\gamma)}=\tilde{g}(t)\overline{\tilde{g}(t-\gamma)},\;\; t\in \mathbb{R}.
\]
Similar to the proof of \cite[Theorem 3.1]{AlaifariRima2021} and \cite[Theorem 4.5]{LiRui2024}, for some $t_0\in \mathbb{R}$,
\[
\tilde{f}(t_0+n\gamma)=e^{j\beta} \tilde{g}(t_0+n\gamma), n\in \mathbb{Z}.
\]
Since $\frac{\Omega^{\prime}}{b}\le\frac{1}{2\gamma}$, $\tilde{f},\tilde{g}\in PW_{\frac{1}{2\gamma}}^2$. Form Proposition~\ref{pro:sha}, we have
$\tilde{f}(t_0+t)=e^{j\beta} \tilde{g}(t_0+t)$ for all $t\in \mathbb{R}.$ Hence, $\tilde{f}=e^{j\beta}\tilde{g}$. Therefore, $f=e^{j\beta}g$, which completes the proof.
\end{proof}
\end{Theorem}

Form Theorem~\ref{Thm:Osam}, we know that a complex-valued function $f\in GPW_{\Omega}^2$ can be uniquely determined by its STOLCT magnitude-only measurements up to a global phase, providing window function satisfies some mild properties.

\begin{Remark}
Our uniqueness results on STOLCT phase retrieval are natural extensions of the corresponding results for STFT \cite{AlaifariRima2021} and STLCT \cite{LiRui2024}.
\end{Remark}

\section{Conclusion}\label{conl}
In this paper, we study the uniqueness of phase retrieval problem form magnitude-only OLCT/STOLCT measurements. First, we discuss the non-trivial ambiguities in continuous OLCT phase retrieval, and show that a compactly supported function can be uniquely recovered from its multiple magnitude-only OLCT measurements. Moreover, we extend the results on non-trivial ambiguities to discrete signals. Finally, for non-separable signals and bandlimited signals, we propose several sufficient and necessary conditions for uniquely reconstructing the original signals from STOLCT amplitude measurements, respectively. In the future, we will still focus on the uniqueness of STOLCT phase retrieval problem for spline functions, or vector functions, and investigate some reconstruction algorithms and convergence analysis for OLCT phase retrieval as well.


 \section*{Acknowledgements}
   This work is supported partially by the National Natural Science Foundation of China (Grant no. 12261059), and the Natural Science Foundation of Jiangxi Province (Grant no. 20224BAB211001).

 \section*{Conflict of Interest}
    This work does not have any conflicts of interest.

 \section*{Data Availability Statement}
    Data sharing is not applicable to this article as the manuscript has no associated data.

\end{document}